\documentclass[letterpaper, 10 pt, conference]{ieeeconf}  

\IEEEoverridecommandlockouts                              
\usepackage{booktabs} 
\usepackage{amsmath}

\usepackage{amssymb, amsthm}
\usepackage{graphicx}
\usepackage{caption, subcaption}
\usepackage{breqn}
\usepackage{array, mathtools}
\usepackage{siunitx}
\usepackage{algorithm}
\usepackage[noend]{algpseudocode}
\usepackage{tabulary}

\newtheorem{proposition}{Proposition}
\newtheorem{theorem}{Theorem}
\theoremstyle{definition}
\newtheorem{Assumption}{Assumption}[section]
\newtheorem{definition}{Definition}

\title{\LARGE \bf Monitor-Based Runtime Assurance for Temporal Logic Specifications}

\author{Matthew Abate, Eric Feron, and Samuel Coogan
\thanks{
This research was supported in part by the National Science Foundation under award \#1749357 and award \#1836932.}
\thanks{M. Abate is with the School of Mechanical Engineering, Georgia Institute of Technology, Atlanta, 30332, USA {\tt\small Matt.Abate@GaTech.edu}.}
\thanks{E. Feron is with the School of Aerospace Engineering, Georgia Institute of Technology, Atlanta, 30332, USA {\tt\small Eric.Feron@Aerospace.GaTech.edu}.}
\thanks{S. Coogan is with the School of Electrical and Computer Engineering and the School of Civil and Environmental Engineering, Georgia Institute of Technology, Atlanta, 30332, USA {\tt\small Sam.Coogan@GaTech.edu}.}
}

\begin{document}

\maketitle
\thispagestyle{empty}
\pagestyle{empty}

\begin{abstract}
This paper introduces the \textit{safety controller} architecture as a runtime assurance mechanism for system specifications expressed as safety properties in Linear Temporal Logic (LTL).
The safety controller has three fundamental components: a performance controller, a backup controller, and an assurance mechanism.
The assurance mechanism uses a \textit{monitor}, constructed as a finite state machine (FSM), to analyze a suggested performance control input and search for system trajectories that are \textit{bad prefixes} of the system specification.
A fault flag from the assurance mechanism denotes a potentially dangerous future system state and triggers a sequence of inputs that is guaranteed to keep the system safe for all time.
A case study is presented which details the construction and implementation of a safety controller on a non-deterministic cyber-physical system.
\end{abstract}

\section{Introduction}
Modern cyber-physical systems (CPS) are complex and sometimes do not behave as expected.
Correctness can be partially addressed with offline verification, however, the end-to-end verification of an entire system is often prevented by its complexity.
In order to compensate for the lack of assurances, it is desirable to enforce correctness online.

For purely cyber systems, runtime correctness can be checked online using monitors \cite{Bartocci2018, Bauer, Neider2018RobustMO, Deshmukh2017, Sayan}. 
In this context, correctness is evaluated with respect to a temporal logic specification; a monitor observes the temporal behaviors of the cyber system and notifies the system operator if a fault is suspected.
Such monitors are convenient for system verification because they can be algorithmically generated from temporal logic specifications.  
Monitors fail, however, to assure correct system behavior; monitors can only detect system faults and do not have the ability to intervene.
As an alternative, edit automata read a sequence of system inputs and, should it recognize a harmful string in that sequence, edits the control input before it is passed on in the software \cite{Ligatti2005}.
Edit automata are formed corresponding to temporal logic specifications, however, it is unclear how hard it is to generate an edit automaton given such a specification.
It was shown by \cite{Bielova}, for instance, that an edit automaton for a given specification is not unique, and that the algorithm presented in \cite{Ligatti2005} for constructing edit automata is not optimal.

For physical control systems, runtime assurance can be performed by identifying a controlled forward invariant region of the state space, and then ensuring that the system does not exceed the boundary of this region. 
Verification techniques of this nature include level set methods \cite{Tomlin}, barrier certificate based methods \cite{Prajna}, and certain model predictive control (MPC) based methods \cite{Wongpiromsarn, Jadbabaie}. 
Computing invariant sets, however, is generally computationally expensive, and moreover, this approach typically is only well suited to analyze invariance conditions or reach conditions.
CPS could, for instance, use an invariance check to avoid stationary obstacles or to patrol a static region of the state space; however, this method is ineffective for more complex system properties, for example, properties which specify that system goals are met in a certain order.  

It is not overtly clear how one might combine results for purely cyber systems with results for physical systems in order to ensure correct behavior of modern CPS against complex specifications. 
For example, by the time a monitor recognizes that a physical system is about to violate its operating specification, the system may have entered an invariant region of the state space from which it is impossible to satisfy the system specification going forward.
A basic idea that has emerged in different contexts is to verify correct system behavior at runtime, and then protect the potentially faulty complex controller with a less sophisticated, but demonstrably safer, backup controller.
This technique is referred to as ``runtime assurance" or the Simplex Architecture \cite{Simplex}.

In this paper, we investigate runtime assurance for system specifications that can be expressed in Linear Temporal Logic (LTL).
In particular, we restrict to the class of \textit{safety} properties.
An LTL property is a safety property if each invalid infinite trace contains a finite prefix which cannot be extended to any valid, satisfying trace \cite{Baier}. 
For instance, completing a set of CPS tasks in a specified order is a safety property, as is returning to a region periodically.
We restrict to safety properties due to certain technical properties of monitors; however, these technical properties correspond to reasonable restrictions on the ability of any runtime assurance mechanism. 
Informally, if an LTL specification is not a safety property, then it contains a liveness condition requiring that some desirable property eventually becomes true. 
For example, reaching a goal region is a liveness property. 
Ensuring that a cyber-physical system does not violate a liveness property appears to be essentially equivalent to solving the controller synthesis problem, which we expressly attempt to avoid in this paper. 
We elaborate on the justification and implications of restricting to safety properties in Section 2.

For system specifications that are expressible as safety properties, we introduce the \textit{safety controller} as a provably correct runtime assurance mechanism.
The safety controller (Figure 1) has three fundamental components: a performance controller, a backup controller and an assurance mechanism.
The assurance mechanism uses a monitor automaton to analyze a suggested performance control input and determine whether the resulting system trajectory is a \textit{bad prefix} for the system specification.
A fault flag from the assurance mechanism denotes a potentially dangerous future system state and triggers an input sequence which is guaranteed to keep the system safe for all time.

\begin{figure}[t!]
    \centering
    \includegraphics[width=.46\textwidth]{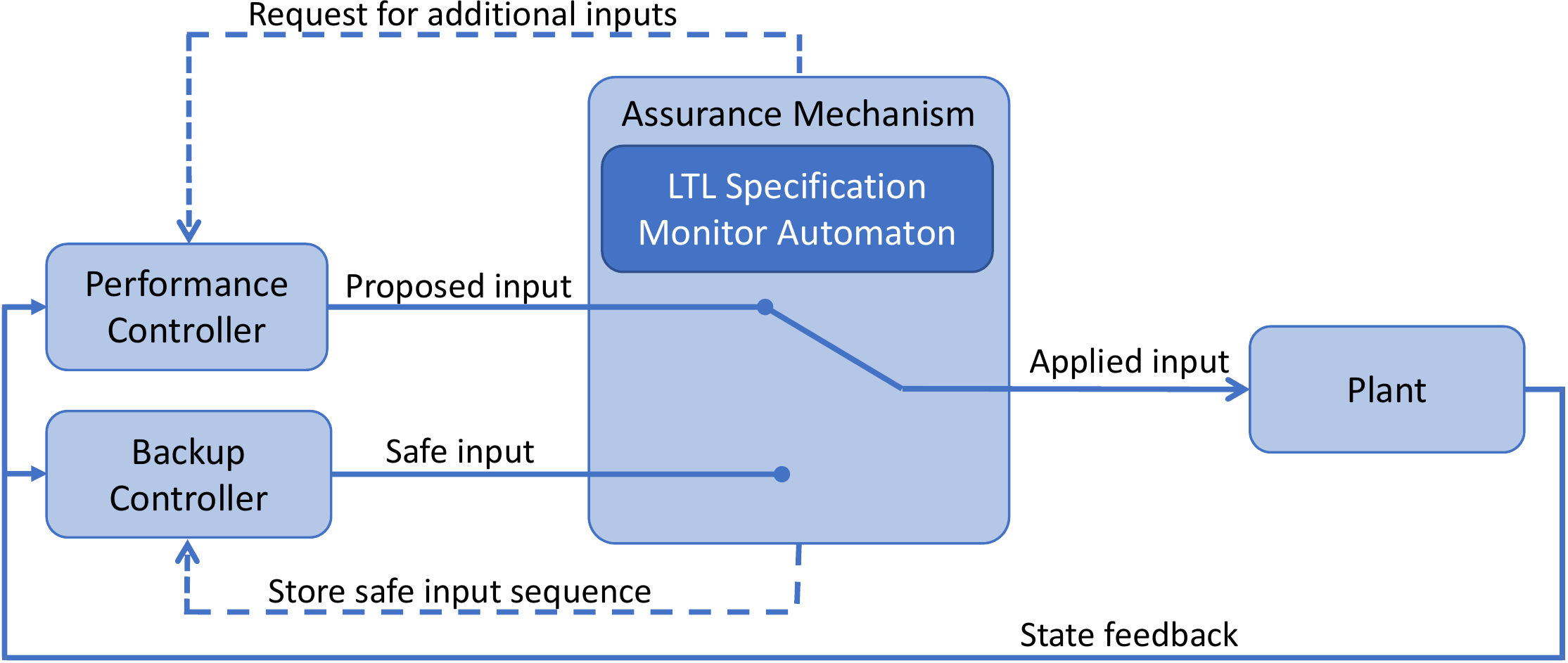}
    \caption{Safety Controller Architecture}
    \label{fig:MS2}
\end{figure}

The backup controller is characterized by a subset of the product of the system and monitor state spaces where correct performance can be provided on an infinite time horizon; we refer to this subset as the high assurance region.
Importantly, the assurance mechanism allows the performance controller to steer the system outside its known high assurance region, provided that the performance controller can demonstrate safety beforehand. 
This creates a trade-off between the \textit{a priori} development of a backup controller and the real-time burden placed on the performance controller. 

This paper is organized in the following way.
We provide a discussion on the construction and application of monitor automata in Section 2.
Section 3 presents the problem formulation and then introduces the safety controller architecture as a correct-by-construction runtime assurance mechanism for deterministic system models. This architecture is expanded in Section 4, to assure non-deterministic CPS system models. 
We provide a discussion on the trade-off between offline platform development and online assurance in Section 5, and the paper concludes with a final case study, provided in Section 6.


\section{Preliminaries on Monitoring for Linear Temporal Logic Properties}
In this section, we introduce \textit{monitor automata} as a runtime verification tool from automata based model checking.


\subsection{Monitor Construction in LTL$_3$}
To formally describe the execution of a CPS, we encode system events as \textit{atomic propositions}; a sequence of events, or a \textit{word}, therefore denotes the progression of the system through time. 
We use $AP$ to denote a set of atomic propositions, $\Sigma = 2^{AP}$ to denote a finite alphabet, and $\Sigma^*$ and $\Sigma^\omega$ to denote the sets of finite and infinite system traces over $\Sigma$, respectively.  
System properties are specified in Linear Temporal Logic (LTL).
For an in depth discussion of the LTL semantics, we refer the reader to \cite{Belta} Section 2.1.

Consider an LTL specification $\varphi$ and a finite path fragment $w \in \Sigma^*$.
Note that there may be no continuations of $w$ which satisfy $\varphi$, or equivalently, the system may have reached a state where it can no longer satisfy its specification under any possible future execution.  
We therefore introduce the definitions of bad and good prefixes in order to describe the set of finite words from which it is impossible to satisfy or violate a system specification.

\begin{definition}[Good and Bad Prefixes]
Consider a finite word $w_1 \in \Sigma^*$ and an LTL property $\varphi$. 
$w_1$ is said to be a \textit{good prefix} for $\varphi$ if $w_1 w_2 \models \varphi$ for all $w_2 \in \Sigma^\omega$.
Similarly, $w_1$ is said to be a \textit{bad prefix} for $\varphi$ if $w_1 w_2 \not\models \varphi$ for all $w_2 \in \Sigma^\omega$.
\end{definition}

\noindent{}Further, we introduce the notion of a \textit{truth value}, in order to classify whether a finite word is a good or bad prefix for $\varphi$.

\begin{definition}[LTL$_3$ Semantics, \cite{Bauer} Section 2.2]
Let $w\in\Sigma^*$ denote a finite word.  The \textit{truth value} of an LTL$_3$ formula $\varphi$ with respect to $w$, denoted $[w\models\varphi]$, is an element of $\mathbb{B}_3 = \{\top, \perp, ?\}$ defined as follows:
\begin{equation*}
[w\models\varphi] = 
\begin{cases}
\top & \text{if } \forall\sigma\in\Sigma^\omega \::\: w\sigma \models \varphi\\
\perp & \text{if } \forall\sigma\in\Sigma^\omega \::\: w\sigma \not\models \varphi\\
? & \text{otherwise.} \\
\end{cases}
\end{equation*}
\end{definition}
\noindent{}Equivalently, the truth value $\varphi$ with respect to $w$ is \textit{true} $``\top"$ if $w$ is a good prefix for $\varphi$, \textit{false} $``\bot"$ if $w$ is a bad prefix for $\varphi$, and \textit{inconclusive} $``?"$ otherwise.

Finally, we formalize the automata-based monitoring procedure for LTL$_3$ (Definition \ref{def:monitor}).  We refer the reader to \cite{Bauer} for a comprehensive discussion on monitor construction, and we provide a sample monitor construction in Figure \ref{samplemonitor}.

\begin{definition}[Monitors in LTL$_3$] \label{def:monitor}
For a given property $\varphi$ a \textit{monitor automaton} $\mathcal{M}^\varphi$ is a finite state machine that reads finite words $w \in \Sigma^*$ and outputs $[w\models \varphi]$.
\end{definition}

\begin{figure}
    \centering
    \includegraphics[width=.40\textwidth]{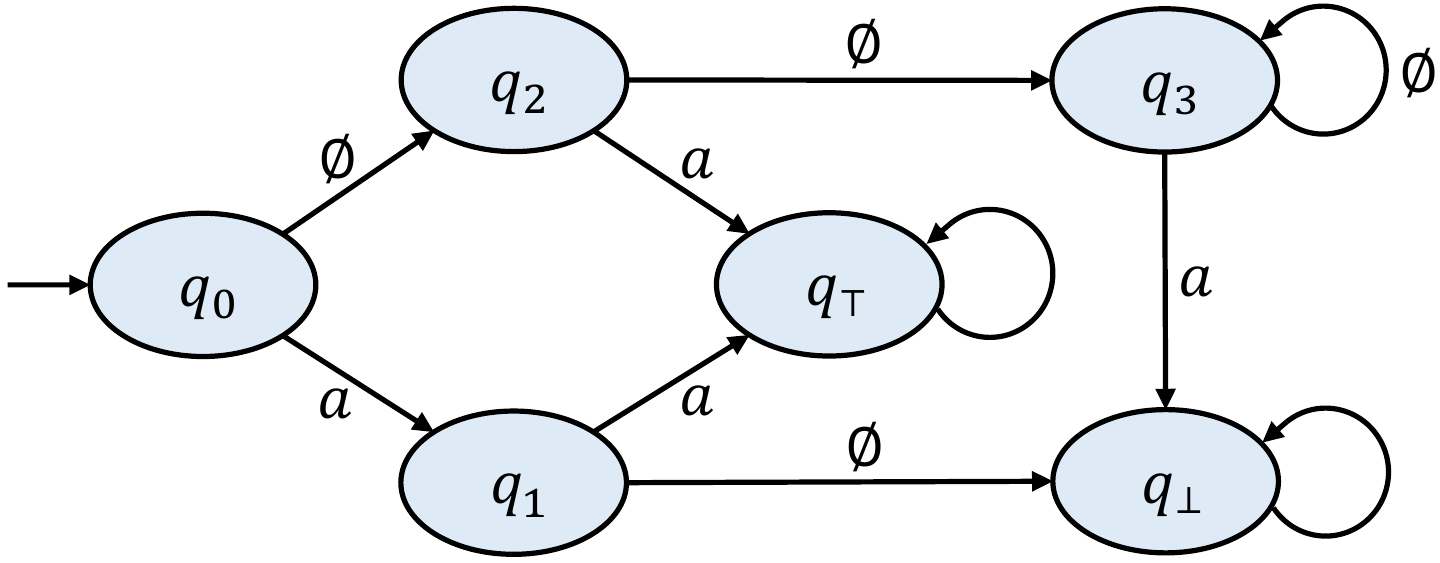}
    \caption{Monitor Automaton $\mathcal{M}^\varphi$ where $\varphi = \square \varnothing \vee \bigcirc a$ is an LTL property evaluated over $\Sigma = \{\varnothing, a\}$.
    Monitor states $q_\top$ and $q_\perp$ output \textit{true} and \textit{false}, respectively, and $q_0, \cdots, q_3$ output \textit{inconclusive}.}
    \label{samplemonitor}
\end{figure}

Note that a monitor state with output $\top$ or $\perp$ will always be a \textit{trap state}, as, for any LTL property, any continuation of a good/bad prefix for that property will also be a good/bad prefix.
Consequently, a monitor $\mathcal{M}^\varphi$ will never have more than one state with output $\top$ and one state with output $\perp$ after minimization.
We therefore use $q_\top$ and $q_\perp$ to denote monitor states with outputs $\top$ and $\perp$, respectively.


\subsection{Monitoring Safety Properties}
Given a LTL property $\varphi$, the monitor automaton $\mathcal{M}^\varphi$ may not be able to reach all of $\mathbb{B}_3$.
For instance, a monitor $\mathcal{M}^\varphi$ for $\varphi = \square a$, $\Sigma = \{\varnothing,\, a\}$, does not include a state $q_\top$.  
We therefore focus the scope of our study to LTL properties which guarantee certain monitor output behaviors.
Specifically, we recall the definition of \textit{safety properties} \cite{Kupferman2001}.

\begin{definition}[Safety Property]\label{def:safetyproperty}
An LTL property $\varphi_{safe}$ is said to be a \textit{safety property} if all words which violate $\varphi_{safe}$ contain a bad prefix for $\varphi_{safe}$.
\end{definition}

LTL safety properties make up a broad class of system specifications; for example, $\varphi = \square \varnothing \vee \bigcirc a$ (Figure \ref{samplemonitor}) is a safety property in LTL.
Furthermore, applying Definition 2.3, we can guarantee system safety by ensuring that the trace of every finite prefix trajectory does not contain a bad prefix for $\varphi_{safe}$.
We formalize this result in Proposition \ref{prop:badprefixes}.

\begin{proposition}\label{prop:badprefixes}
An infinite word $w \in\Sigma^\omega$ satisfies $\varphi_{safe}$ if $[\hat{w} \models \varphi_{safe}] \neq \perp$ for all prefixes $\hat{w} \in \text{Pref}(w)$.
\end{proposition}

Note that a monitor $\mathcal{M}^{\varphi_{safe}}$, for a safety property $\varphi_{safe}$ not semantically equal to \textit{true}, is guaranteed to have a state with output $\perp$.  
Moreover, we can verify system trajectories against $\varphi_{safe}$ by ensuring the system run over $\mathcal{M}^{\varphi_{safe}}$ never enters $q_\bot$. 


\section{The Safety Controller Architecture}
Here we apply monitor automata in a framework which \textit{enforces} system objectives at runtime.  
For system objectives which are expressible as safety properties in LTL, we create a notion of system safety and describe a control architecture which is guaranteed to keep the system safe for all time.
We refer to this control architecture as a \textit{safety controller}.
The safety controller has three fundamental components: a high-performance controller, a backup controller, and an assurance mechanism. 
The assurance mechanism is implemented through two algorithms which ensure system safety by continually choosing between the high-performance control input and the high-assurance control input.


\subsection{Preliminaries and Notation}
Consider a discrete-time non-deterministic control system
\begin{equation}\label{eq:nondeterministic}
    x^+ = f(x, u, d)
\end{equation}
\noindent{}where, $x_k \in \mathcal{X} \subseteq \mathbb{R}^n$, $u_k \in \mathcal{U} \subseteq \mathbb{R}^m$ and $d_k \in \mathcal{D} \subset \mathbb{R}^p$ represent the system state, control input, and a non-deterministic bounded disturbance, at a time $k \in \mathbb{N}_{\geq 0}$, respectively.
We make no assumptions on the stochastic properties of $d_k$ within $\mathcal{D}$, i.e. $d_k$ is not assumed to conform to any probability distribution. 

Associated with the system (\ref{eq:nondeterministic}) is a set of atomic propositions $AP$ and a labeling function $L:\mathcal{X}\rightarrow{}2^{AP}$.  
As before, let $\Sigma = 2^{AP}$ denote a finite alphabet, and let $\Sigma^*$ and $\Sigma^\omega$ denote the sets of finite and infinite traces of the system (\ref{eq:deterministic}) over $\Sigma$, respectively.
Additionally, let $\varphi$ be property in LTL.
We aim to construct a string of inputs $u_0, u_1, \cdots$ which, when fed to the system, will create a system trace $L(x_0\, x_1 \cdots ) \models\varphi$.


\subsection{Assurance Through Monitoring}
We first consider the case with no disturbance $d$.
After developing a general framework for approaching this deterministic formulation, we return to the non-deterministic setting of (\ref{eq:nondeterministic}).
To that end, consider a system
\begin{equation}\label{eq:deterministic}
    x^+ = f(x, u)
\end{equation}
\begin{equation*}
    x_k \in \mathcal{X} \subseteq \mathbb{R}^n, \quad u_k \in \mathcal{U} \subseteq \mathbb{R}^m, \quad k \in \mathbb{N}.
\end{equation*}
\noindent{}We aim to ensure that a safety property $\varphi \in \text{LTL}$ is satisfied over a (infinite) system trace.  
Following Proposition \ref{prop:badprefixes}, we say that a finite system trajectory $x_0, \cdots{}, x_k$ is \textit{safe} if $[L(x_0 \cdots x_k) \models \varphi] \in \{\top, ?\}$, and \textit{unsafe} if $[L(x_0 \cdots x_k) \models \varphi] = \; \perp$.
A controller which is implemented on the system must first ensure the system stays safe before pursuing any auxiliary system objectives, such as enforcing optimality constraints or attempting to satisfy additional LTL properties.

Assume we have access to a controller which is claimed to be able to keep the system safe, while also possibly meeting some auxiliary system objective.  
We denote this a \textit{performance controller}. 
Here, the performance controller is characterized by three assumptions:
\begin{Assumption}
At each time $k$, the performance controller proposes a control input to be applied to the system.
\end{Assumption}
\begin{Assumption}
If requested, the performance controller can propose future inputs $\tilde{u}_{k}, \cdots, \tilde{u}_{k+N_{max}}$, for some $N_{max}\in\mathbb{N}$, such that at time $j$, $k \leq j \leq k+N_{max}$ the performance controller intends to apply $\tilde{u}_j$ to the system.
The performance controller is not obligated to use these control inputs at future time steps.
\end{Assumption}
\begin{Assumption}
The performance controller might bear faulty, i.e. applying the performance control input to (\ref{eq:deterministic}) may eventually cause the system to violate its specification $\varphi$.
\end{Assumption}

This definition of the performance controller is quite broad and encompasses general unverified feedback control laws.
For example, a human operator, whose effectiveness cannot be verified \textit{a priori}, can be thought of as a performance controller for a manned CPS.
Despite the lack of global assurances, we assume that certain aspects of the performance control law make it preferable if safety can be determined beforehand. 
For example, the performance control may be designed to optimize some objective or achieve some auxiliary goal.

We additionally assume the existence of a controller which is known to satisfy the safety objective, but might be limited in abilities or performance. 
We denote this mechanism a \textit{backup controller}, formalized next.

\begin{definition}[Backup Controller and High Assurance Region]\label{def:backitup}
For a system of the form (1), a safety property $\varphi \in \text{LTL}$, and a corresponding monitor $\mathcal{M}^{\varphi} = (\Sigma,\, Q,\, q_0,\, \delta,\, \lambda)$, let $S = \mathcal{X}\times Q$ denote the total state space of the combined system and monitor.  
A \textit{backup controller} is characterized by a subset of the state space $S^b \subseteq S$ such that for any $(x,q)\in S^b$, there exists an infinite sequence of control inputs known to the backup controller such that the resulting infinite horizon system trace satisfies $\varphi$.  
$S^b$ is referred to as the \textit{high assurance region} of the backup controller.
\end{definition}

Let us fix a safety property $\varphi$, a corresponding monitor $\mathcal{M}^\varphi = (\Sigma,\, Q,\, q_0,\, \delta,\, \lambda)$, a backup controller and a high assurance region $S^b$ for the remainder of this section.  
We recall the problem formulation: given a system of the form (\ref{eq:deterministic}), generate a sequence of control inputs which, when fed to the system, will create a system trace which satisfies $\varphi$.
Trivially, the system can be kept safe for all time by applying the backup control input to the system at each time step.
This method of input generation, however, restricts the system to operating inside the high assurance region, and removes the possibility of reaching any auxiliary goals.
For this reason, we interpret our problem formulation as follows: create a logical architecture which chooses whether to apply the performance control input or the backup control input to the system, at each time step, such that the resulting system trace satisfies $\varphi$.
In this case, whenever possible, the chosen input should be that of the performance controller.
We refer to this embedded decision logic as the \textit{assurance mechanism} of the safety controller.
Importantly, we give the system the ability to leave $S^b$ by applying the performance control input to the system in instances where the assurance mechanism can provide assurance in the form of a suggested input sequence, certifying that the system can return to the high assurance region once it has left.
We justify this decision with Proposition \ref{prop:wayhome}. 

\begin{proposition}\label{prop:wayhome}
Let $x_0, \cdots, x_k$ denote a finite trajectory of system (\ref{eq:nondeterministic}), and let  $q_0, \cdots, q_k$ be the corresponding run of the trajectory over $\mathcal{M}^\varphi$.  If $(x_k, q_k)\in S^b$, then $L(x_0 \cdots x_k)$ is not a bad prefix for $\varphi$.
\end{proposition}

It follows from Proposition \ref{prop:wayhome} that if the performance controller is able to suggest a path back to the high assurance region $S^b$ from every point along a system trajectory, then the performance control input is safe.
We therefore create a procedure which takes an initial state $(x,q)\in \mathcal{X}\times Q$ and returns a finite sequence of control inputs $\mu_r$ such that:
\begin{enumerate}
    \item $\mu_r = \varnothing$ signifies a fault in the performance controller, and
    \item $\mu_r = u_0, \cdots, u_k$ signifies that applying the sequence of control inputs will return the system to $S^b$.
\end{enumerate}
\noindent{}For an initial system monitor state, we use the term \textit{recovery input sequence} to denote such a sequence of control inputs which drive the system state into $S^b$.

We implement this procedure by calling \textsc{Recovery}$(x, q)$, provided in Algorithm 1.
Succinctly, \textsc{Recovery}$(x, q)$ simulates the system dynamics using the current system monitor state and a non-empty sequence of suggested future performance control inputs; if the future system monitor state is contained in $S^b$, then the suggested input sequence is returned, and, if not, the null sequence is returned, signaling a system fault.
Therefore, if the performance controller is not faulty, then the output of Algorithm 1 is a non-empty recovery input sequence for the current system monitor state.
\begin{algorithm}[t!]
\caption{Generate A Non-Empty Recovery Input Sequence for the Current State}
\begin{algorithmic}[1]
\setlength\tabcolsep{0pt}
\Statex\begin{tabulary}{\linewidth}{@{}LLp{6cm}@{}}
\textbf{input}&:\:\:& a current state $(x, q)\in S$\\
\textbf{output}&:\:\:& a sequence of control inputs $\mu_r$
\end{tabulary}

\Function{Recovery}{$x, q$}
\State \textbf{Initialize: } $\tilde{x}_0 = x$, $\tilde{q}_0 = q$, $i=0$
\While{$(i \leq N_{max})$ \textbf{and} $(\tilde{q}_i \neq q_{_\perp})$}
  \State $\tilde{u}_i \gets \text{Request\textunderscore{}Next\textunderscore{}Input}$
  \State $\tilde{x}_{i+1} \gets f(\tilde{x}_i, \tilde{u}_i)$
  \State $\tilde{q}_{i+1}\gets\delta(\tilde{q}_i, L(\tilde{x}_{i+1}))$
  \If{$(\tilde{x}_{i+1}, \tilde{q}_{i+1}) \in S^b$}
    \State\textbf{return} $\:\mu_r = \tilde{u}_0,\,\cdots,\,\tilde{u}_i $
  \EndIf
  \State $i \gets i+1$
\EndWhile
\State\textbf{return} $\:\mu_r = \varnothing$
\EndFunction
\State\textbf{end function}
\end{algorithmic}
\end{algorithm}
We use Algorithm 1 to design a logical architecture for our safety controller as follows:
\begin{enumerate}
    \item At each time $k$ the decision logic calls $\textsc{Recovery}(x_k, q_k)$ and stores the output in a variable $\mu_r$.
    \item If $\mu_r$ is non-empty, then $u_0$ is applied to the system and $u_1, \cdots, u_j$ is stored to memory.
    \item If $\mu_r = \varnothing$ then the performance controller has suggested an input sequence that cannot be verified. The memorized recovery input sequence is applied to return the system to $S^b$, and the backup control input is then applied for all future time.
\end{enumerate}
We implement this procedure through Algorithm 2.
Algorithms 2 calls Algorithm 1 once per discrete time step, prompting the performance controller for a string of future inputs which steer the system to $S^b$.
The performance controller will be considered functional if such a sequence is found.
In the case that the performance controller is functional, the aforementioned procedure will create a finite series of inputs known to create a safe trajectory, but then only apply the first input to the system.
In this way, the method resembles a model predictive control framework.
Note that the assurance mechanism will need to simulate the system dynamics as many as $N_{max} + 1$ times per discrete system timestep, therefore it is assumed that the assurance mechanism has the necessary functionality to perform these computations.
Choosing control inputs using Algorithm 2 guarantees system safety for all time.
We formalize this result in Theorem 1.
\begin{algorithm}
\caption{Runtime Assurance for Discrete-Time Control Systems}
\begin{algorithmic}[1]
\State $\mu_r \gets \textsc{Recovery}(x_0, q_0)$
\While{$(\mu_r \neq \varnothing)$}
\State \textbf{Apply} $u_0$ to system
\State \textbf{Store} $u_1, \cdots, u_j$ to \textbf{Memory}
\State $(x, q) \gets \textsc{Current\_State}$
\State $\mu_r \gets \textsc{Recovery}(x, q)$
\EndWhile
\State \textbf{Apply} Recovery Input Sequence to system
\While{(\textbf{True})}
\State \textbf{Apply} \textsc{Backup\_Input} to system
\EndWhile
\end{algorithmic}
\end{algorithm}
\begin{theorem}[Runtime Assurance for Deterministic CPS]\label{safetytheorem}
Let $\varphi \in \text{LTL}$ be a safety property, and let $w = L(x_0\,x_1 \cdots )$ be the trace of an infinite trajectory resulting from a sequence of inputs chosen using Algorithm 2.  If the system is initialized in the high assurance region, i.e. $(x_0, q_0)\in S^b$, then $w \models \varphi$.
\end{theorem}

In Section 4, we expand upon the results of this section in order to design a runtime assurance procedure for non-deterministic systems.
To that end, we omit a proof for Theorem \ref{safetytheorem}, with the knowledge that a later theorem encapsulates the findings of this section (Theorem \ref{NDtheorem}).



\section{Runtime Assurance for Non-Deterministic CPS}
In practice, it is important to consider CPS models which account for potential system disturbances.  We therefore turn our study toward assuring for non-deterministic systems.


\subsection{Preliminaries and Notation}
Reconsider the discrete-time non-deterministic CPS model (\ref{eq:nondeterministic}).  
We retain the definitions of the system parameters $x$, $u$ and $d$, as well as the associated sets $AP$, $\Sigma$, $\Sigma^*$, and $\Sigma^\omega$, from Section 3.
Additionally let $\varphi$ denote a mission objective that is a safety property in LTL.


\subsection{Runtime Assurance in the Presence of Disturbances}
We aim to extend our results from Section 3 to assure the non-deterministic system (\ref{eq:nondeterministic}) against  $\varphi$.  
To that end, we retain the definition of a \textit{backup controller} from Section 2, now applied to (\ref{eq:nondeterministic}), and we use the term \textit{performance controller} to denote a controller which claims it will keep the system safe while also meeting some auxiliary objective. 
Here, we take the following assumption on the performance control law, replacing Assumptions III.2:
\begin{Assumption}
If requested, the performance controller can also propose a sequence $(X_k, g_k),$ $\cdots, (X_{k+N_{max}}, g_{k+N_{max}})$ of feedback control laws $g_i: \mathcal{X} \rightarrow \mathcal{U}$ and regions of the state space $X_i\subseteq \mathcal{X}$, such that at a future time $i$, the performance controller intends to choose its inputs using $g_i$ provided the current state of the system $x_i\in X_i$.
The performance controller is not obligated to choose future inputs using these policies.
\end{Assumption}
\noindent{}Note that Assumption IV.1 subsumes Assumption III.2 as a special case. 
In particular, the potential sequence of proposed inputs in Assumption III.2 can be considered a certain static feedback map $g_i$ such that $g_i(x)=u_i$ for all $x$.

For an initial state $(x_0, q_0) \in S$, and a sequence of feedback control laws $g_0, g_1, \cdots$, we denote the set of reachable states after $i$ steps by $R_i$, i.e.
\begin{equation*}
    \begin{split}
        R_1 & = \{(x, q) \mid x = f(x_0, g_0(x_0), d), \\ 
        & \qquad\qquad\qquad\qquad\qquad d\in \mathcal{D}, \: q = \delta(q_0, L(x))\} \\
        R_{i+1} & = \{(x, q) \mid x = f(\bar{x}, g_{i}(\bar{x}), d), \: (\bar{x}, \bar{q})\in R_{i}, \\
        & \qquad\qquad\qquad\qquad\qquad d\in \mathcal{D},\: q = \delta(\bar{q}, L(x))\}
    \end{split}
\end{equation*}
\noindent{}for all $i>1$.
While calculating reachable sets can be computationally expensive, numerous methods exist to approximate reachable sets efficiently \cite{Kurzhanski1993}. 
Therefore, let $\tilde{R}_i$ denote an over-approximation of $R_i$, i.e. $R_i \subseteq \tilde{R}_i$.

Using the approach described in Section 3 as a framework, we create a procedure which takes a current state $(x, q)\in S$ and returns a sequence of feedback control law $\gamma_r$ such that:
\begin{enumerate}
    \item $\gamma_r = \varnothing$ signifies a fault in the performance controller, and
    \item $\gamma_r = g_0, \cdots, g_{k}$ signifies that choosing inputs according to the sequence of feedback control laws sequential will return the system to $S^b$.
\end{enumerate}
\noindent{}We implement this procedure by calling a function
\textsc{RecoveryD}$(x, q)$ (Algorithm 3).
In the non-deterministic setting, we redefine the term \textit{recovery input sequence} to denote a sequence of feedback control laws which is guaranteed to steer the system to a state contained in $S^b$.  
Therefore, if the performance controller is not faulty, then the output of Algorithm 3 can be thought of as a non-empty recovery input sequence for the current state.

\begin{algorithm}
\caption{Generate A Non-Empty Recovery Input Sequence in the Presence of Disturbances}
\begin{algorithmic}[1]
\setlength\tabcolsep{0pt}
\Statex\begin{tabulary}{\linewidth}{@{}LLp{6cm}@{}}
\textbf{input}&:\:\:& a current state $(x, q)\in S$.\\
\textbf{output}&:\:\:& a sequence of feedback control laws $\gamma_r$
\end{tabulary}
\Function{RecoveryD}{$x, q$}
\State \textbf{Initialize: } $\tilde{R}_0 = \{(x, q)\}$, $i=0$
\While{$(i\leq N_{max})$}
\State $X_i, g_i\gets$ Request\_Next\_Control\_Law
\If{$\{x\mid (x, q)\in \tilde{R}_i\} \not\subseteq X_i$}
  \State \textbf{return} $\gamma_r = \varnothing$
\EndIf
\State \textbf{Compute} $\tilde{R}_{i+1}$ using $\tilde{R}_{i}$ and $g_i$
\If{$(\tilde{R}_{i+1}\subseteq S^b)$}
  \State \textbf{return} $\gamma_r = g_0,\, \cdots,\, g_i$
\EndIf
\State $i\gets i+1$
\EndWhile
\State \textbf{return} $\gamma_r = \varnothing$\
\EndFunction
\State\textbf{end function}
\end{algorithmic}
\end{algorithm}

We use Algorithm 3 to design a logical architecture for our safety controller as follows:
\begin{enumerate}
    \item At each time $k$ the decision logic calls \textsc{RecoveryD}$(x_k, q_k)$ and stores the output in a variable $\gamma_r$.
    \item If $\gamma_r$ is non-empty, then $g_0(x_k)$ is applied to the system and $g_1, \cdots, g_j$ is stored to memory.
    \item If $\gamma_r = \varnothing$, then the performance controller has suggested an input sequence that cannot be verified. The memorized recovery input sequence is applied to return the system to $S^b$, and the backup control input is then applied for all future time.
\end{enumerate}
\noindent{}This procedure is implemented using Algorithm 4.

\begin{algorithm}
\caption{Runtime Assurance for Non-Deterministic Discrete-Time Control Systems}
\begin{algorithmic}[1]
\State \textbf{Initialize: } $x = x_0$, $q = q_0$
\State $\gamma_r \gets \textsc{RecoveryD}(x, q)$
\While{$(\gamma_r \neq \varnothing)$}
\State \textbf{Apply} $g_0(x)$ to system; 
\State \textbf{Store} $g_1, \cdots, g_j$ to \textbf{Memory}
\State $(x, q) \gets \textsc{Current\_State}$
\State $\gamma_r \gets \textsc{RecoveryD}(x, q)$
\EndWhile
\State \textbf{Apply} Recovery Input Sequence to system
\While{(\textbf{True})}
\State \textbf{Apply} \textsc{Backup\_Input} to system
\EndWhile
\end{algorithmic}
\end{algorithm}

\begin{theorem}[Runtime Assurance for Non-Deterministic CPS]\label{NDtheorem}
Let $\varphi \in \text{LTL}$ be a safety property, and let $w = L(x_0\,x_1 \cdots )$ be the trace of an infinite trajectory resulting from a sequence of inputs chosen using Algorithm 4.  If the system is initialized in the high assurance region, i.e. $(x_0, q_0)\in S^b$, then $w \models \varphi$.
\end{theorem}
\begin{proof}
Let $(x_0, q_0)\in S^b$ be the initial system monitor state. Therefore $q_0\neq q_\perp$.
The performance controller is prompted for a non-empty recovery input sequence $\gamma_r$ for the current state (Line 2, Line 7). 
If $\gamma_r = g_0,\,\cdots,\,g_N$ is returned (Line 3), then the $u_0 = g_0(x_0)$ is passed to the system, which brings the system to a current state $(x_{1},\,q_{1})$, and $g_1,\,\cdots,\,g_N$ is saved to memory as a recovery input sequence for $(x_{1},\, q_{1})$.
The existence of a recovery input sequence for $(x_{1},\, q_{1})$ guarantees $q_{1}\neq q_\perp$. It follows inductively that if the performance controller is always able to supply a recovery input sequence for the current state, then $L(x_0\,x_1 \cdots ) \models \varphi$.
If instead there exists a finite time $k$ such that the performance controller is unable to generate a recovery input sequence for $(x_{k},\,q_{k})$ then the previous recovery input sequence will be applied to the system and then the backup control input is applied for all time (Line 8 and 9).  This forces the system to remain in $S^b$ indefinitely and therefore guarantees system safety by Proposition 2.
\qed
\end{proof}

Note that when $\mathcal{D} = \varnothing$, the non-deterministic system (\ref{eq:nondeterministic}) is equivalent to the deterministic system (\ref{eq:deterministic}).  Therefore, Theorem 1 is a special case of Theorem 2.

\section{Classifying the Trade-off Between A Priori Development and Effectiveness}
Here, we discuss the feasibility of designing and implementing a safety controller architecture. 
Specifically, we address the offline computation involved in constructing a  backup controller and the computational complexity of the algorithms which run online in the assurance mechanism.

\subsection{Developing a Safety Controller Offline}
Two steps are required to construct a safety controller from an LTL safety specification $\varphi$ and a given system model: the first step is to convert the system specification to a monitor automaton and the next step is to design a backup controller.

We refer the reader to \cite{Bauer} for a detailed discussion on monitor construction.
In short, this process involves realizing two deterministic finite automata (DFAs) and then computing their minimal product automaton.
There are well-defined procedures for constructing a DFA from an LTL specification, as well as procedures for computing product automata \cite{Baier}.
The resulting DFA is reduced to its minimal form by removing every unreachable state and every pair of  non-distinguishable states.
This procedure is implemented using Moore's algorithm \cite{minimization}, which has an average complexity $O(n\cdot{}\text{log}(n))$ when minimizing a DFA with $n$ states.
This procedure will only need to be performed once in the CPS development process.

A backup control law is developed by identifying a controlled invariant region of the system monitor state space $S^b \subseteq \mathcal{X}\times Q$, such that for all $(x, q) \in S^b$, $q\neq q_\perp$.
This follows from Proposition \ref{prop:wayhome}, i.e. in order to ensure system safety the backup controller must only ensure the system trace never enters the false monitor state $q_\perp$.
In order to compute such an invariant region of the system monitor state space, abstraction based methods are possible \cite{Belta, Tabuada}.


\subsection{Online Computation for the Assurance Mechanism}
Consider a CPS outfitted with a safety controller architecture.
We would like to characterize the computational resources which will be necessary online for assurance.

First, we propose that Assumption III.2 (IV.2), which requires the performance controller to suggest potential future inputs, fits well with existing control architectures; for instance, all MPC controllers employ this functionality. Moreover, since the performance controller is assumed to be provided in advance of the safety controller construction process, we do not discuss the computational complexity of designing a performance controller for a given CPS. Instead, let us discuss the requisite online computation requirements of the assurance mechanism and the backup controller.

The assurance mechanism will call Algorithm 1  (or 3) once per discrete time step, in order to access the performance control input. Therefore, the assurance mechanism must have sufficient computational capabilities to over-approximate reachable system states at runtime. Reachable set computation is a very well-studied problem in the controls and hybrid systems literature, with numerous efficient algorithms. See \cite[Chapter 29]{levine2010control} for an overview.


\subsection{Discussion}
The assurance mechanism will only ever apply a performance control input sequence if the set of reachable states resulting from that sequence is contained inside the high assurance region. There is, therefore, an intuitive trade-off between the \textit{a priori} development of the assurance mechanism and the amount of computational resources which are necessary online. For example, the likelihood of finding a safe performance control input sequence is maximized for safety controllers with large, well developed, high assurance regions.  Similarly, an assurance mechanism which computes tight approximations of reachable sets will be less likely to return a fault flag, in comparison to an assurance mechanism which uses conservative approximations.  In general, therefore, developing a comprehensive backup control policy, will reduce the burden placed on the assurance mechanism at runtime. This constitutes a trade-off between the \textit{a priori} development of the assurance mechanism and the effectiveness of the safety controller architecture.

Here we include some additional methods for increasing the capabilities of the assurance mechanism. In the instance that the system leaves $S^b$, with assurance from the performance controller, an assurance mechanism might choose to expand $S^b$ on the fly with the knowledge of the new system state and its corresponding recovery input sequence. As the high assurance region expands, the system will effectively \textit{learn} how to stay safe in previously unverified product states, equivalently increasing system performance.  
Safe learning is an active research area in the controls and formal methods communities. One modern technique, referred to as shielding, enforces LTL safety properties in a run-time assurance framework \cite{AAAI1817211}. Shields, however, only assure discrete-time discrete-state systems, whereas the safety controller architecture is applicable to systems with a continuous state space.
Shielding also assumes a probabilistic system disturbance, which we expressly attempt to avoid in our construction.
Additionally, in the instance that the performance controller is determined to be faulty, Algorithms 2 and 4 first apply the recovery input sequence to the system, forcing the system monitor trajectory to reenter $S^b$, and then pass the backup control input to the system for all time; an assurance mechanism may instead re-activate the preforming controller once $S^b$ is reached. In this case, the performance controller is certainly faulty and may again suggest an unsafe sequence of inputs at some-point in the future, however, even in this instance, the safety controller architecture will be able to assure the system against the mission objective.



\section{Case Study: An Accelerating DeLorean}
A safety controller architecture was implemented on a modified F1/10 race car (Figure \ref{fig:car}). F1/10 is an open-source 1/10 scale autonomous vehicle test-bed designed primarily for use by academic researchers \cite{foneten}.

\subsection{System Characterization}
When traveling along a straight line, the plant dynamics of the system conform to a non-deterministic discrete-time double integrator model.  Here the system state at a time $k\in \mathbb{N}_{\geq 0}$, is described by the car's position along the road $x(k) \in \mathbb{R}_{\geq 0}$ and the car's forward velocity $v(k) \in \mathbb{R}_{\geq 0}$. 
Control inputs were suggested by a human operator, who chose the applied motor torque with a wireless Logitech Gamepad F710 controller (Figure \ref{fig:car}); this input is therefore proportional to the experienced acceleration.  Additionally, a non-deterministic factor was included in the system model in order to encapsulate the effects of drag on the vehicle.

\subsection{Generating a Safety Controller given a Mission Objective}
Our mission objective is taken from the movie \textit{Back to the Future}: when the car passes the clock tower, the car's velocity must be greater than 2 meters per second. We give the car, hereafter referred to as a DeLorean, an initial position $(x_0, v_0) = (0, 0)$, and arbitrarily place the clock tower a distance 2.54 meters from the origin.

Without an enforcement mechanism, the human operator has the ability to suggest an input sequence which causes the vehicle to pass the clock tower with inadequate speed, thus violating the mission objective.
We therefore design an assurance mechanism to act as a filter between the the human operator and the plant, and ensure that the runtime performance of the DeLorean meets the mission objective.

\begin{figure}
    \centering
    \includegraphics[width=.265\textwidth]{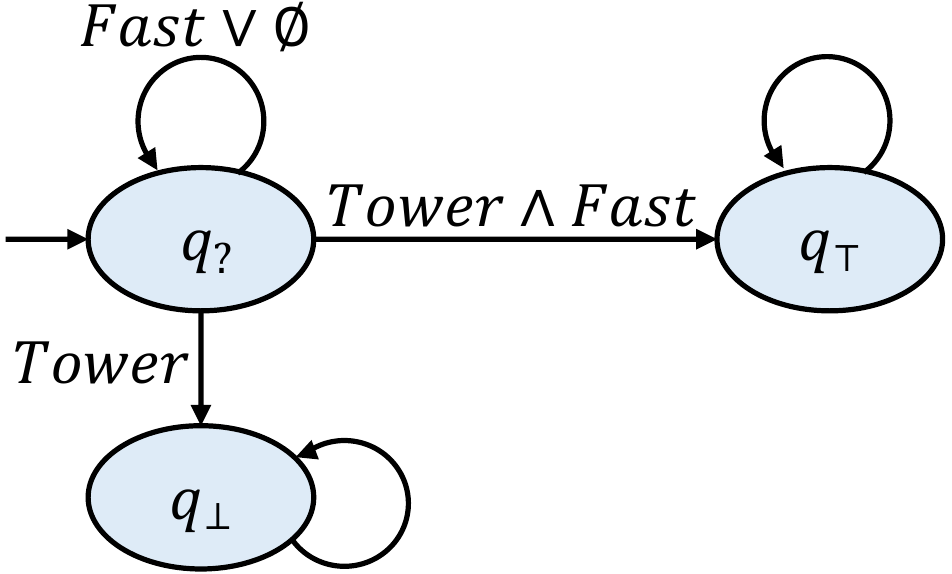}
    \caption{Monitor $\mathcal{M}^\varphi$ for $\varphi = (\neg Tower) \: U \:(Tower \wedge Fast)$.  Monitor state $q_?$ outputs inconclusive.}
    \label{fig:bttfmonitor}
\end{figure}

We convert our system specification to an LTL safety property as follows.
Let $AP = \{Tower, Fast\}$ be the set of events, where $Tower$ indicates that the DeLorean has driven past the clock tower, and $Fast$ indicates that the velocity of the DeLorean is greater than 2 m/s.
The trace of a system trajectory is therefore given by the labeling function  $L: \mathbb{R}^2 \rightarrow 2^{AP}$
\begin{equation*}
    L(x, v) = 
    \begin{cases}
        \varnothing & x < 2.54, \: v < 2\\
        Tower & x \geq 2.54, \: v< 2\\
        Fast & x < 2.54, \: v \geq 2\\
        Tower \wedge Fast & x \geq 2.54, \: v \geq 2.
    \end{cases}
\end{equation*}
\noindent{}In order to ensure that the DeLorean does not pass the clock tower with insufficient velocity, we enforce the LTL safety specification $\varphi = (\neg Tower) \: U \: (Tower \wedge Fast)$ whose monitor $\mathcal{M}^\varphi$ is provided in Figure \ref{fig:bttfmonitor}.
We take $S^b$ to be a triangular region of the state space from which the DeLorean can decelerate to zero velocity safely:
\begin{equation*}
    S^b = \{(x, v, q) \mid q = q_{\top} \:\textbf{ or }\: v \leq -0.69x+1.66, \: q = q_?\}.
\end{equation*}
\noindent{}In this case, if the DeLorean is in a current state $(x, v, q)\in S^b$, then the backup controller will suggest that the DeLorean brake such that the vehicle decelerates to a stop. 
A safety controller architecture is created by integrating Algorithms 3 and 4 into an assurance mechanism.

\begin{figure}[t!]
    \centering
    \begin{subfigure}{.45\textwidth}
        \centering
        \includegraphics[width = .8\textwidth]{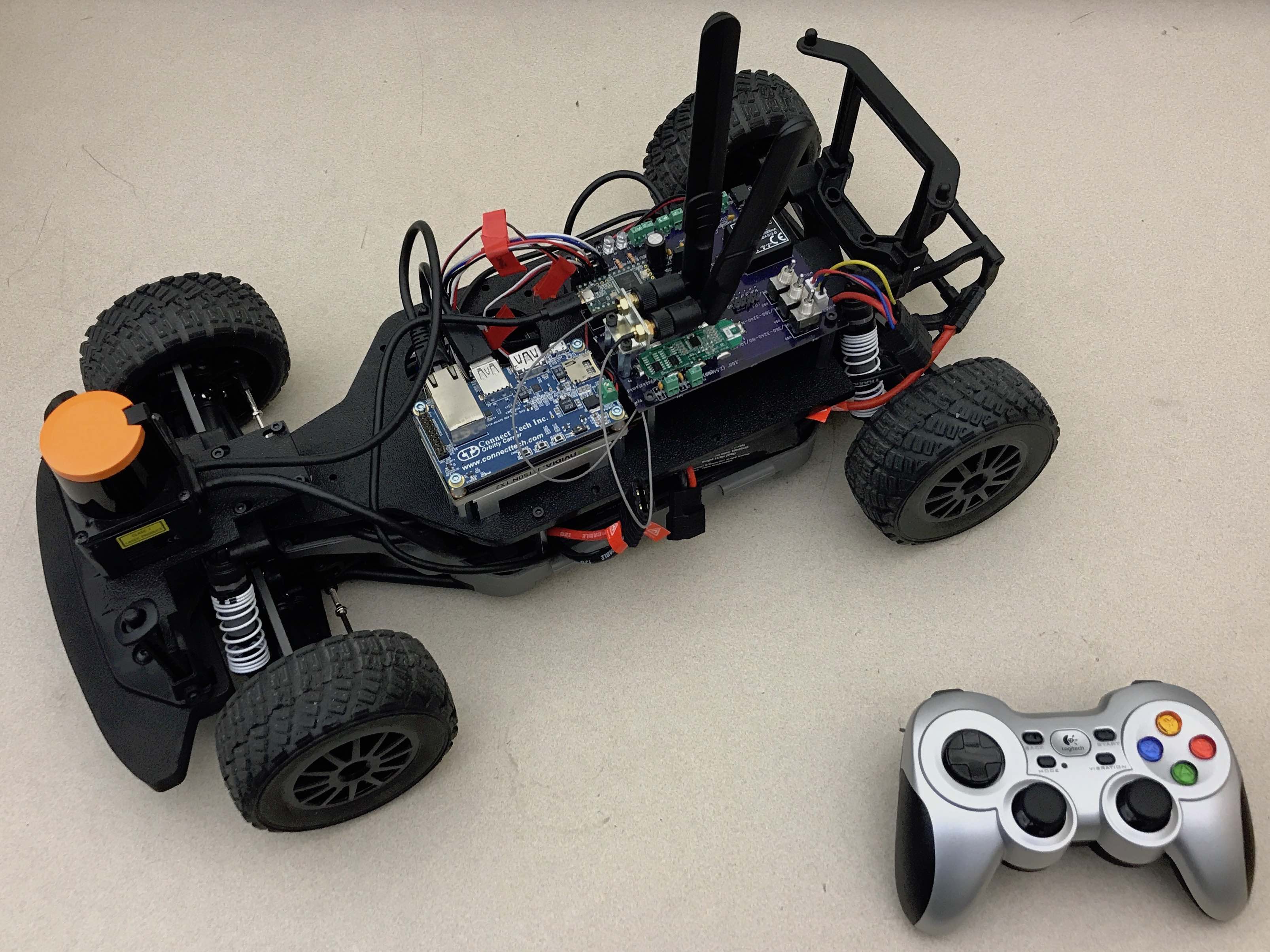}        
        \caption{Experimental test-bed.  A human operator suggested performance control inputs to an modified F1/10 race car (right) using a Logitech Gamepad F710 controller (lower left).}
        \label{fig:car}
    \end{subfigure}
    \par\bigskip
    \begin{subfigure}{.45\textwidth}
        \centering
        \includegraphics[width = .85\textwidth]{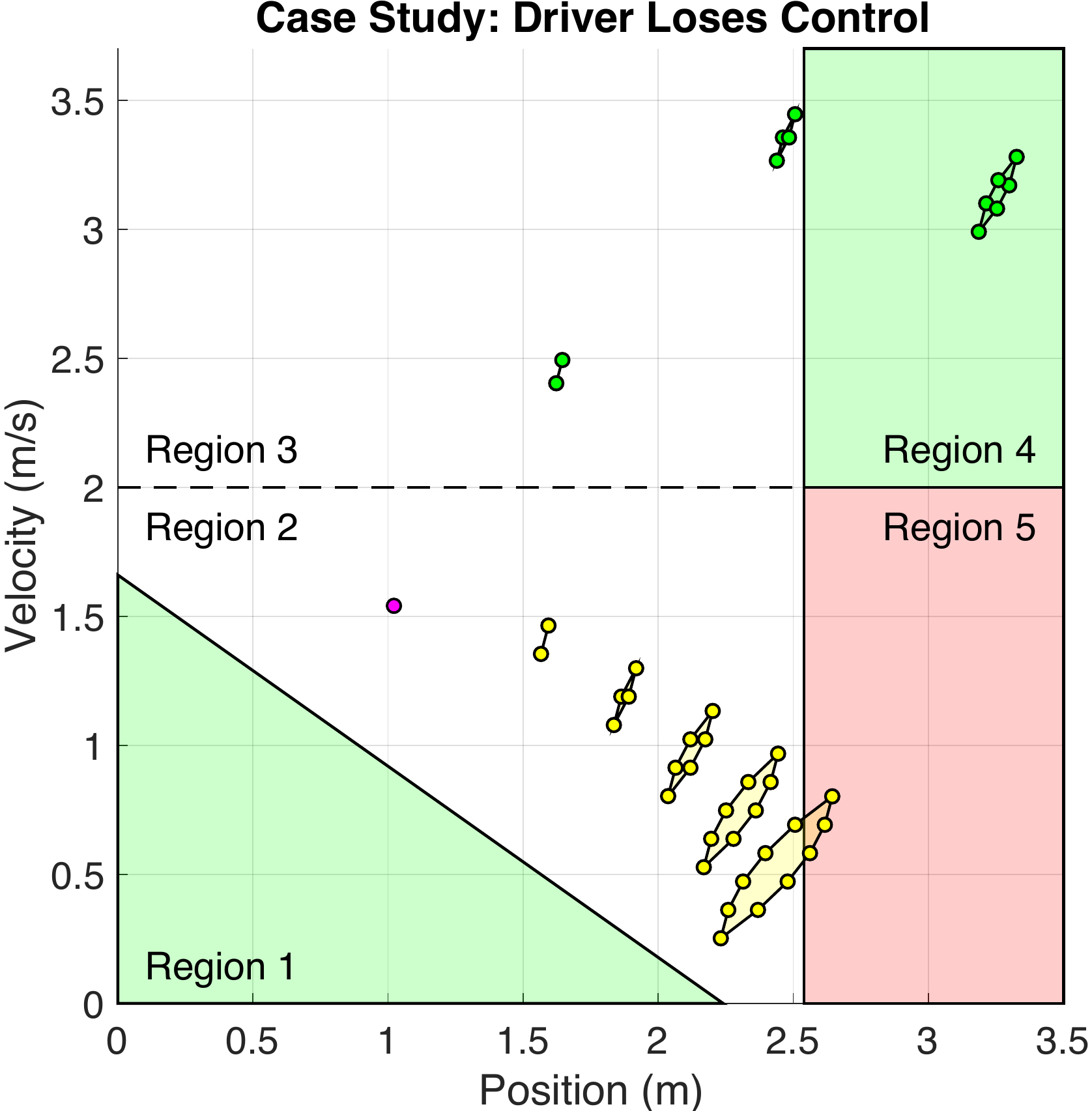}
        \caption{Safety controller implementation on F1/10 race car.  The safe zone of the backup controller, $S^b$ is shown in green (Regions 1 and 4). Regions 2, 3, 4 and 5 are labeled $\varnothing$, $Fast$, $Tower \wedge Fast$ and  $Tower$, respectively. The current vehicle state is shown in pink.  The non-deterministic trajectory resulting from a memorized recovery input sequence is shown in green, and the driver's suggested trajectory, which causes the system to violate the mission objective, is shown in yellow.}
        \label{fig:other}
    \end{subfigure}
    \caption{Case study test-bed and trial data.}
\end{figure}

\subsection{Case Study Implementation}
We present the scenario where the vehicle driver, who initially suggested safe inputs, suggests an unsafe control policy (Figure \ref{fig:other}). Performance control inputs are passed to the system 250 milliseconds after the driver sent them via remote control; this lag-time allowed the assurance mechanism to analyze control inputs as though they were suggested as a string. 
The driver first suggested an input sequence that guaranteed the that DeLorean would satisfy the mission objective $\varphi$.  This allows the DeLorean to leave $S^b$.  At a future timestep, the driver suggests a control input sequence which allowed for the possibility that the DeLorean would violate $\varphi$.
The assurance mechanism then applies the memorized recovery input sequence, and the DeLorean passes the clock tower with sufficient velocity.


\section{Acknowledgements}
The authors wish to thank Will Stuckey for his work with test-bed development.

\section{Conclusions}
This paper introduces the safety controller as a runtime assurance mechanism for system objectives expressed in linear temporal logic.
A case study is presented which details the construction and implementation of a safety controller on a non-deterministic cyber-physical system.

\bibliographystyle{ieeetr}
\bibliography{Monitor.bib}

\end{document}